\documentclass[a4paper]{article}

\usepackage{a4wide}
\usepackage[utf8]{inputenc}
\usepackage[english]{babel}
\usepackage{hyperref}
\usepackage{amssymb}
\usepackage{tikz}
 \usepackage[hmargin=2.5cm,vmargin=3cm]{geometry}
\usepackage[color=green!30]{todonotes}
\usepackage[framemethod=TikZ]{mdframed}
\usepackage{apxproof}

\usepackage{enumitem}

\mdfdefinestyle{MyFrame}{
	linecolor=black,
	innerbottommargin=\baselineskip,
	backgroundcolor=gray!10!white}

\newtheorem{theorem}{Theorem}

\newtheorem{proposition}[theorem]{Proposition}

\newtheorem{lemma}[theorem]{Lemma}
\newtheorem{observation}[theorem]{Observation}
\newtheorem{corollary}[theorem]{Corollary}

\newtheorem{conjecture}[theorem]{Conjecture}

\newcommand{\geod}[1]{g\left(#1\right)}
\newcommand{\dem}[1]{\mathrm{dem}\left(#1\right)}
\newcommand{\meg}[1]{\mathrm{meg}\left(#1\right)}
\newcommand{\cyclo}[1] {c\left(#1\right)}
\newcommand{\leaf}[1]{\ell\left(#1\right)}
\newcommand{\ipec}[1]{\mathrm{isopec}\left(#1\right)}
\newcommand{\ipp}[1]{\mathrm{isopp}\left(#1\right)}

\newcommand{\mdim}{\mathrm{mdim}}  
\newcommand{\edim}{\mathrm{edim}}  
\newcommand{\dist}[2]{d\left(#1,#2\right)}
\newcommand{\legs}[1] {l\left(#1\right)}
\newcommand{\branchres}[1] {L\left(#1\right)}

\newcommand{\ff}[1]{{\color{black} #1}}
\newcommand{\ah}[1]{{\color{black} #1}}

\begin{document}

\title{Distance-based (and path-based) covering problems for graphs of given cyclomatic number\footnote{The second and third authors were supported by the IDEX-ISITE initiative CAP 20-25 (ANR-16-IDEX-0001), the International Research Center ``Innovation Transportation and Production Systems'' of the I-SITE CAP 20-25, and the ANR project GRALMECO (ANR-21-CE48-0004). The third author was supported by the Jenny and Antti Wihuri Foundation, the Research Council of Finland grant number 338797.}}


\author{Dibyayan Chakraborty\footnote{\noindent School of Computer Science, University of Leeds, United Kingdom.}
\and Florent Foucaud\footnote{\noindent Université Clermont Auvergne, CNRS, Clermont Auvergne INP, Mines Saint-\'Etienne, LIMOS, 63000 Clermont-Ferrand, France.}
\and Anni Hakanen\footnote{\noindent Turku Collegium for Science, Medicine and Technology, University of Turku, Finland.}~\footnote{\noindent Department of Mathematics and Statistics, University of Turku, FI-20014, Finland.}~\footnotemark[3]}
%
%

%
\maketitle

\begin{abstract}

We study a large family of graph covering problems, whose definitions rely on distances, for graphs of bounded cyclomatic number (that is, the minimum number of edges that need to be removed from the graph to destroy all cycles). These problems include (but are not restricted to) three families of problems: (i) variants of metric dimension, where one wants to choose a small  set $S$ of vertices of the graph such that every vertex is uniquely determined by its ordered vector of distances to the vertices of $S$; (ii) variants of geodetic sets, where one wants to select a small set $S$ of vertices such that any vertex lies on some shortest path between two vertices of $S$; (iii) variants of path covers, where one wants to select a small set of paths such that every vertex or edge belongs to one of the paths. We generalize and/or improve previous results in the area which show that the optimal values for these problems can be upper-bounded by a linear function of the cyclomatic number and the degree~1-vertices of the graph. To this end, we develop and enhance a technique recently introduced in [C. Lu, Q. Ye, C. Zhu. Algorithmic aspect on the minimum (weighted) doubly resolving set problem of graphs, \emph{Journal of Combinatorial Optimization} 44:2029--2039, 2022] and give near-optimal bounds in several cases. This solves (in some cases fully, in some cases partially) some conjectures and open questions from the literature. The method, based on breadth-first search, is of algorithmic nature and thus, all the constructions can be computed in linear time. {Our results also imply an algorithmic consequence for the computation of the \emph{optimal} solutions: for some of the problems, they can be computed in polynomial time for graphs of bounded cyclomatic number.}
\end{abstract}

\section{Introduction}

Distance-based covering problems in graphs are a central class of problems in graphs, both from a structural and from an algorithmic point of view, with numerous applications. Our aim is to study such problems for graphs of bounded cyclomatic number. The latter counts the number of edges that need to be removed from a graph to make it acyclic; it is a measure of sparsity of the graph that is popular in both structural and algorithmic graph theory~\cite{CV85,KN12,whitney1931non} that has classic applications in program testing~\cite{cyclomatic_complexity}.

Although distance-based covering problems are very diverse, they share some properties that make them behave similarly in certain contexts, a fact that has already been observed for some of these problems, in particular, the metric dimension and the geodetic number problems~\cite{BDM24,FGKLMST2024b,FGKLMST2024a}. Notably, regarding algorithmic applications, contrary to more ``locally defined'' problems, they often do not behave well for graphs of bounded treewidth~\cite{KLP22,KeKo20,LP22}. The goal of this paper is to demonstrate that, for the more restrictive graphs of bounded cyclomatic number, interesting bounds can be derived, using a similar technique that is both simple and powerful. The obtained bounds also lead to efficient algorithms for these graphs. We will mainly focus on three types of such problems, as follows.

\medskip

\paragraph{Metric dimension and its variants.} \sloppy In these concepts, introduced in the 1970s~\cite{S:leavesTree,Harary76}, the aim is to distinguish elements in a graph by using distances. A set $S \subseteq V(G)$ is a \emph{resolving set} of $G$ if for all distinct vertices $x,y \in V(G)$ there exists $s \in S$ such that $\dist{s}{x} \neq \dist{s}{y}$. The smallest possible size of a resolving set of $G$ is called the \emph{metric dimension} of $G$ {(denoted by $\dim (G)$)}. 
During the last two decades, many variants of resolving sets and metric dimension have been introduced. In addition to the original metric dimension, we consider the edge and mixed metric dimensions of graphs. A set $S \subseteq V(G)$ is an \emph{edge resolving set} of $G$ if for all distinct edges $x,y \in E(G)$ there exists $s \in S$ such that $\dist{s}{x} \neq \dist{s}{y}$, where the distance from a vertex $v$ to an edge $e = e_1e_2$ is defined as $\min \{\dist{v}{e_1},\dist{v}{e_2}\}$~\cite{KelencEdge18}. A mixed resolving set is both a resolving set and an edge resolving set, but it must also distinguish vertices from edges and vice versa; a set $S \subseteq V(G)$ is a \emph{mixed resolving set} of $G$ if for all distinct $x,y \in V(G) \cup E(G)$ there exists $s \in S$ such that $\dist{s}{x} \neq \dist{s}{y}$~\cite{KelencMixed}. The \emph{edge metric dimension} {$\edim (G)$} (resp. \emph{mixed metric dimension} {$\mdim (G)$}) is the smallest size of an edge resolving set (resp. mixed resolving set) of $G$. More on the different variants of metric dimension and their applications (such as detection problems in networks, graph isomorphism, coin-weighing problems or machine learning) can be found in the recent surveys~\cite{KuziakYeroSurvey,TillquistSurvey}.

\medskip

\paragraph{Geodetic numbers.} A \emph{geodetic set} of a graph $G$ is a set $S$ of vertices such that any vertex of $G$ lies on some shortest path between two vertices of $S$~\cite{GS}. The \emph{geodetic number} $\geod{G}$ of $G$ is the smallest possible size of a geodetic set of $G$.

The version where the edges must be covered is called an \emph{edge-geodetic set}~\cite{EGS}. ``Strong'' versions of these notions have been studied. A \emph{strong (edge-) geodetic set} of graph $G$ is a set $S$ of vertices of $G$ such that we can assign for any pair $x,y$ of vertices of $S$ a shortest $xy$-path such that each vertex (edge) of $G$ lies on one of the chosen paths~\cite{SGS,preprintBresil,MKXAT17}. 

Recently, the concept of \emph{monitoring edge-geodetic set} was introduced in~\cite{MEG1} (see also~\cite{MEGbilo2024inapproximability,dev2023monitoring,MEG-CALDAM2024,MEG2}) as a strengthening of a strong edge-geodetic set: here, for every edge $e$, there must exist \ah{two vertices $x,y$ in the monitoring edge-geodetic set} such that $e$ lies on \emph{all} shortest paths between $x$ and $y$. 

These concepts have numerous applications related to the field of convexity in graphs, see the book~\cite{bookGC}.

We also consider the concept of \emph{distance-edge-monitoring-sets} introduced in~\cite{DEM1,DEM1conf} and further studied in~\cite{JKLMZ24,YKMD24,YYHMK24}, which can be seen as a relaxation of monitoring edge-geodetic sets. A set $S$ is a distance-edge-monitoring-set if, for every edge $e$ of $G$, there is a vertex $x$ of $S$ and a vertex $y$ of $G$ such that $e$ lies on all shortest paths between $x$ and $y$.

\medskip

\paragraph{Path covering and partition problems.} In this type of problem, one wishes to cover the vertices (or edges) of a graph using a small number of paths. For path partition problems, the paths are required to be vertex-disjoint, but for path covering, they may not be. A \emph{path cover} (\emph{path partition}, respectively) is a set of (vertex-disjoint) paths of a graph $G$ such that every vertex of $G$ belongs to one of the paths. If one path suffices, the graph is said to be Hamiltonian, and deciding this property is one of the most fundamental graph-algorithmic problems. The paths may be required to be shortest paths, in which case we have the notion of an \emph{isometric path cover/partition}~\cite{chakraborty2023isometric,FloManuscrit,DFPT24,PPP,IPC,TG21}; if they are required to be chordless, we have an \emph{induced path cover/partition}~\cite{PPP,LLM03,manuel2018revisiting}. The edge-covering versions have also been studied~\cite{ANDREATTA95}. This type of problems has numerous applications, such as program and circuit testing~\cite{ANDREATTA95,NH79}, vehicle routing~\cite{bus} or bioinformatics~\cite{dagPC}. They are the subject of well-known studies in graph theory, such as the Gallai-Milgram theorem~\cite{GM60} or conjectures by Berge~\cite{berge1983path} and Gallai~\cite{Lovasz1968}.



\medskip 
\paragraph{Our goal.} Our objective is to study the three above classes of problems, on graphs of bounded cyclomatic number. (See Figure~\ref{fig:diagram} for a diagram showing the relationships between the optimal solution sizes of the studied problems.) A \emph{feedback edge set} of a graph $G$ is a set of edges whose removal turns $G$ into a forest. The smallest size of such a set, denoted by $\cyclo{G}$, is the \emph{cyclomatic number} of $G$~\cite{berge1973graphs}. It is sometimes called the \emph{feedback edge (set) number} or the \emph{cycle rank} of $G$. For a graph $G$ on $n$ vertices, $m$ edges and $k$ connected components, it is not difficult to see that we have $\cyclo{G}=m-n+k$, since a forest on $n$ vertices with $k$ components has $n-k$ edges. In this paper, we assume all our graphs to be connected. To find an optimal feedback edge set of a connected graph, it suffices to consider a spanning tree; the edges not belonging to the spanning tree form a minimum-size feedback edge set. 

Graphs whose cyclomatic number is constant have a relatively simple structure. They are sparse (in the sense that they have a linear number of edges). They also have bounded treewidth (indeed the treewidth is at most the cyclomatic number plus one), a parameter that plays a central role in the area of graph algorithms, see for example Courcelle's celebrated theorem~\cite{COURCELLE}. Thus, they are studied extensively from the perspective of algorithms (for example for the metric dimension~\cite{EpsteinWeighted}, the geodetic number~\cite{KeKo20} or other graph problems~\cite{CV85,DEHJLUU19,FKK01,KN12,UW13}). In particular, in many cases, distance-related problems are computationally hard on graphs of bounded treewidth~\cite{KLP19,KLP22,KeKo20,LP22}. Thus, for this type of problems, it is of interest to design efficient algorithms for graphs of bounded cyclomatic number. Graphs of given cyclomatic number are also studied from a more structural angle~\cite{alcon23,SedlarMixedED,SedlarCyclo2,SedlarCacti,whitney1931non}.


\paragraph{Conjectures addressed in this paper.} 
In order to formally present the conjectures, we need to introduce some structural concepts and notations. A \emph{leaf} of a graph $G$ is a vertex of degree~1, and the number of leaves of $G$ is denoted by $\leaf{G}$. Consider a vertex $v \in V(G)$ of degree at least~3. A \emph{leg} attached to the vertex $v$ is a path $p_1 \ldots p_k$ such that $p_1$ is adjacent to $v$, $\deg (p_k) = 1$ and $\deg (p_i) = 2$ for all $i \neq k$. The number of legs attached to the vertex $v$ is denoted by $\legs{v}$.

A set $R \subseteq V(G)$ is a \emph{branch-resolving set} of $G$, if for every vertex $v \in V(G)$ of degree at least 3 the set $R$ contains at least one element from at least $\legs{v} - 1$ legs attached to $v$. The minimum cardinality of a branch-resolving set of $G$ is denoted by $\branchres{G}$, and we have
\[ \branchres{G} = \sum_{v \in V(G), \, \deg(v) \geq 3, \, \legs{v} > 1} (\legs{v} - 1). \]

It is well-known that for any tree $T$ with at least one vertex of degree~3, we have 
$\dim(T)=\branchres{T}$ (and if $T$ is a path, then $\dim(T)=1$)~\cite{Chartrand00,Harary76,Khuller96,S:leavesTree}. This has motivated the following conjecture.

\begin{conjecture}[\cite{SedlarCacti}]\label{conj-dimedim}
	 Let $G$ be a connected graph with $\cyclo{G} \geq 2$. Then $\dim (G) \leq \branchres{G} + 2 \cyclo{G}$ and $ \edim (G) \leq \branchres{G} + 2 \cyclo{G}$.
\end{conjecture}

The restriction $\cyclo{G} \geq 2$ is missing from the original formulation of Conjecture~\ref{conj-dimedim} in~\cite{SedlarCacti}. However, Sedlar and \v{S}krekovski have communicated to us that this restriction should be included in the conjecture. Conjecture~\ref{conj-dimedim} holds for cacti with $\cyclo{G} \geq 2$~\cite{SedlarCacti}.
The bound $\dim(G) \leq \branchres{G} + 18\cyclo{G} - 18$ was shown in \cite{EpsteinWeighted} (for $\cyclo{G}\geq 2$), and is the first bound established for the metric dimension in terms of $\branchres{G}$ and $\cyclo{G}$ (note that the bound holds even for a vertex-weighted variant of metric dimension). The bound $\dim (G) \leq \branchres{G}+6c(G)$ was proved in~\cite{BousquetMdSparseZero}.

\ah{Sedlar and \v{S}krekovski \cite{SedlarCyclo2} also posed the following refinement of the previous conjecture, where $\delta (G)$ is the minimum degree of~$G$.}

\begin{conjecture}[\cite{SedlarCyclo2}]\label{conj-mindeg2}
	 If $\delta (G) \geq 2$ and $G \neq C_n$, then $\dim (G) \leq 2 \cyclo{G} - 1$ and $\edim (G) \leq 2 \cyclo{G} - 1$.
\end{conjecture}

In \cite{SedlarCyclo2}, Sedlar and \v{S}krekovski showed that Conjecture~\ref{conj-mindeg2} holds for graphs with minimum degree at least 3. They also showed that if Conjecture~\ref{conj-mindeg2} holds for all 2-connected graphs, then it holds for all graphs $G$ with $\delta (G) \geq 2$. 
Recently, Lu at al.~\cite{LuWeightedDoubly} addressed Conjecture~\ref{conj-mindeg2} and showed that $\dim (G) \leq 2c(G) + 1$ when $G$ has minimum degree at least 2.

\begin{conjecture}[\cite{SedlarMixedED}]\label{conj-mdim}
	 Let $G$ be a connected graph. If $G \neq C_n$, then $\mdim (G) \leq \leaf{G} + 2 \cyclo{G}$.
\end{conjecture}

Conjecture~\ref{conj-mdim} is known to hold for trees~\cite{KelencMixed}, cacti and 3-connected graphs~\cite{SedlarMixedED}, and balanced theta graphs~\cite{SedlarMdimCyclomatic}.

The following conjecture on distance-edge-monitoring sets was also posed recently.

\begin{conjecture}[\cite{DEM1,DEM1conf}]\label{conj-dem}
	 For any graph $G$, $\dem{G}\leq\cyclo{G}+1$.
\end{conjecture}
\textcolor{black}{Here, $\dem{G}$ is the the smallest size of a distance-edge-monitoring set of graph $G$.}
The original authors of the conjecture proved the bound when $\cyclo{G}\leq 2$, and proved that the bound $\dem{G}\leq2\cyclo{G}-2$ holds when $\cyclo{G}\geq 3$~\cite{DEM1}. The conjectured bound would be tight~\cite{DEM1,DEM1conf}.

Regarding monitoring edge-geodetic sets, although no formal conjecture was explicitly posed, the bound $\meg{G}\leq 9\cyclo{G}+\leaf{G}-8$ was proved in~\cite{dev2023monitoring,MEG1}, and it was asked whether this can be improved.

\begin{figure}[t]
\centering
\scalebox{0.7}{\begin{tikzpicture}[node distance=7mm]

\tikzstyle{mybox}=[fill=white,line width=0.5mm,rectangle, minimum height=.8cm,fill=white!70,rounded corners=1mm,draw];
\tikzstyle{myedge}=[line width=0.5mm,->]
\newcommand{\tworows}[2]{\begin{tabular}{c}{#1}\\{#2}\end{tabular}}

    \node[mybox,fill=gray!20] (meg) {\tworows{monitoring}{edge-geodetic number}};
    \node[mybox,fill=gray!20] (dem)  [below left=of meg] {\tworows{distance}{edge-monitoring number}}  edge[myedge] (meg);
    \node[mybox,fill=gray!20] (isopec) [below right=of meg,xshift=1mm] {\tworows{isometric path}{edge-cover number}};
    \node[mybox,fill=gray!20] (isopp) [right=of isopec,xshift=35mm] {\tworows{isometric path}{partition number}};
    \node[mybox] (isopc) [right=of isopec,yshift=-10mm] {\tworows{isometric path}{cover number}} edge[myedge] (isopec) edge[myedge] (isopp);
    \node[mybox] (indpec) [below=of isopec, yshift=-20mm] {\tworows{induced path}{edge-cover number}} edge[myedge] (isopec);
    \node[mybox] (indpp) [below=of isopp,yshift=-10mm] {\tworows{induced path}{partition number}} edge[myedge] (isopp);
    \node[mybox] (indpc) [below=of isopc,yshift=-20mm] {\tworows{induced path}{cover number}} edge[myedge] (isopc) edge[myedge] (indpec) edge[myedge] (indpp);
    \node[mybox] (pec) [below=of indpec,yshift=-10mm] {\tworows{path}{edge-cover number}} edge[myedge] (indpec);
    \node[mybox] (pp) [below=of indpp,yshift=-10mm] {\tworows{path}{partition number}} edge[myedge] (indpp);
    \node[mybox] (pc) [below=of indpc,yshift=-10mm] {\tworows{path}{cover number}} edge[myedge] (indpc) edge[myedge] (pec) edge[myedge] (pp);
    
    \node[mybox] (segs) [left=of isopec,yshift=-10mm] {\tworows{strong edge-geodetic}{number}} edge[myedge] (isopec) edge[myedge] (meg);
    \node[mybox] (sgs) [below=of segs,yshift=-5mm] {\tworows{strong geodetic}{number}} edge[myedge] (isopc) edge[myedge] (segs);
    \node[mybox] (egs) [left=of segs,yshift=-10mm] {\tworows{edge-geodetic}{number}} edge[myedge] (segs);
    \node[mybox,fill=gray!20] (gs) [below=of egs,yshift=-5mm] {\tworows{geodetic}{number}} edge[myedge] (egs) edge[myedge] (sgs);

    \node[mybox,fill=gray!20] (mmd) [below=of gs,xshift=20mm] {\tworows{mixed metric}{dimension}};


    \node[mybox,fill=gray!20] (md) [below=of mmd,xshift=-15mm] {\tworows{metric}{dimension}} edge[myedge] (mmd);
    \node[mybox,fill=gray!20] (emd) [below=of mmd,xshift=15mm] {\tworows{edge metric}{dimension}} edge[myedge] (mmd);

  \end{tikzpicture}}

\caption{Relations between the graph parameters discussed in the paper. If a parameter $A$ has a directed path to parameter $B$, then for any graph, the value of $A$ is upper-bounded by a linear function of the value of $B$. The problems in gray boxes are explicitly studied in this paper; bounds for the other problems follow from these results.}
\label{fig:diagram}
\end{figure}
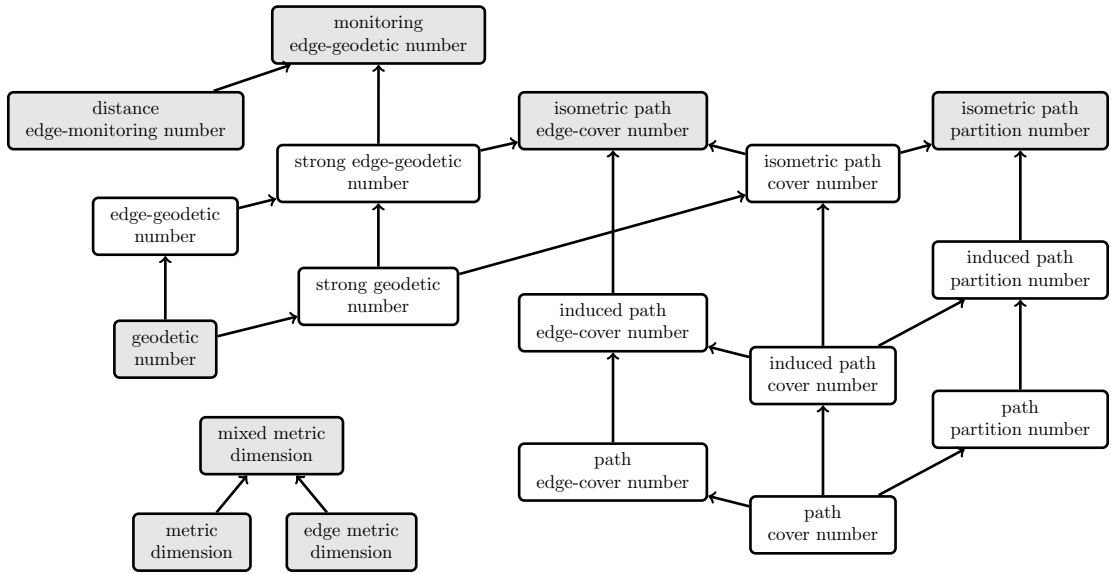


\paragraph{Our contributions.} In this paper, we are motivated by Conjectures~\ref{conj-dimedim}-\ref{conj-dem}, which we address. 
We will show that both $\dim (G)$ and $\edim (G)$ are bounded from above by $\branchres{G} + 2 \cyclo{G} + 1$ for all connected graphs $G$. Moreover, we show that if $\branchres{G} \neq 0$, then the bounds of Conjecture~\ref{conj-dimedim} hold. 
 
We show that Conjecture~\ref{conj-mdim} is true when $\delta (G) = 1$, and when $\delta (G) \geq 2$ and $G$ contains a cut-vertex. We also show that $\mdim (G) \leq 2 \cyclo{G} + 1$ in all other cases. We also consider the first part of Conjecture~\ref{conj-dimedim}, that $\dim (G) \leq \branchres{G} + 2 \cyclo{G}$ from~\cite{SedlarCacti}, in the case where $\delta (G) = 1$, and we show that it is true when $\branchres{G} \geq 1$ and otherwise we have $\dim (G) \leq 2 \cyclo{G} + 1$. We also consider the conjecture that $\edim (G) \leq \branchres{G} + 2 \cyclo{G}$ from~\cite{SedlarCacti}, and we show that it is true when $\delta (G) = 1$ and $\branchres{G} \geq 1$, and when $\delta (G) \geq 2$ and $G$ contains a cut-vertex. We also show that $\edim (G) \leq 2 \cyclo{G} + 1$ in all other cases. 
 
 Thus, our results yield significant improvements towards the Conjectures~\ref{conj-dimedim}-\ref{conj-mdim}, since they are shown to be true in most cases, and are approximated by an additive term of 1 for all graphs.
 
 Moreover, we also resolve in the affirmative Conjecture~\ref{conj-dem}.

To obtain the above results, we develop a technique from~\cite{LuWeightedDoubly}, who introduced it in order to study a strengthening of metric dimension called \emph{doubly resolving sets} in the context of graphs of minimum degree~2. We notice that the technique can be adapted to work for all graphs and in fact it applies to many types of problems: (variants of) metric dimension, (variants of) geodetic sets, and path-covering problems. For all these problems, the technique yields upper bounds of the form $a\cdot\cyclo{G}+f(\ell(G))$, where $\ell(G)$ is the number of leaves of $G$, $f$ is a linear function that depends on the respective problem, and $a$ is a small constant.

The technique is based on a breadth-first-search rooted at a specific vertex, that enables to compute an optimal feedback edge set $F$ by considering the edges of the graph that are not part of the breadth-first-search spanning tree. We then select vertices of the edges of $F$ (or neighbouring vertices); the way to select these vertices depends on the problem. For the metric dimension and path-covering problems, a pre-processing is done to handle the leaves of the graph (for the geodetic set variants, all leaves must be part of the solution). Our results demonstrate that the techniques used by most previous works to handle graphs of bounded cyclomatic number were not precise enough, and the simple technique we employ is much more effective. We believe that it can be used with sucess in similar contexts in the future.

A preliminary version of this paper (without most of the proofs) appeared in the proceedings of the FCT 2023 conference~\cite{CFH23}.

\paragraph{Algorithmic applications.}
For all the considered problems, our method in fact implies that the optimal solutions can be computed in polynomial time for graphs with bounded cyclomatic number. In other words, we obtain XP algorithms with respect to the cyclomatic number. This was already observed in~\cite{EpsteinWeighted} for the metric dimension (thanks to our improved bounds, we now obtain a better running time, however it should be noted that in~\cite{EpsteinWeighted} the more general weighted version of the problem was considered).

\paragraph{Organisation.} We first describe the general method to compute the special feedback edge set in Section~\ref{sec:method}. We then use it in Section~\ref{sec:md} for the metric dimension and its variants. We then turn to geodetic sets and its variants in Section~\ref{sec:geod}, and to path-covering problems in Section~\ref{sec:paths}. We describe the algorithmic consequence in Section~\ref{sec:algo}, and conclude in Section~\ref{sec:conclu}.

\section{The general method}\label{sec:method}


The \emph{length} of a path $P$, denoted by $|P|$, is the number of its vertices minus  one. A path is \emph{induced} if there are no graph edges joining non-consecutive vertices. A path is \emph{isometric} if it is a shortest path between its endpoints.  For two vertices $u,v$ of a graph $G$, $\dist{u}{v}$ denotes the length of an isometric path between $u$ and $v$. 
Let $r$ be a vertex of $G$. An edge $e=uv\in E(G)$ is a \emph{horizontal} edge \emph{with respect to} $r$ if $\dist{u}{r} = \dist{v}{r}$ {(otherwise, it is a \emph{vertical edge})}. For a vertex $u$ of $G$, let $B_r(u)$ denote the set of edges $uv\in E(G)$ such that  $\dist{u}{r} = \dist{v}{r} + 1$, \ah{i.e., the vertical edges where one endpoint is $u$ and the other endpoint $v$ is closer to $r$ than $u$}. A set $F$ of edges of $G$ is \emph{good with respect to} $r$ if $F$ contains all horizontal edges with respect to $r$ and for each $u\neq r$, $|B_r(u) \cap F| = |B_r(u)| - 1$. A set $F$ of edges is simply \emph{good} if $F$ is good with respect to some vertex $r\in V(G)$. 
For a set $F$ of  good edges of a graph $G$, let $T_F$ denote the subgraph of $G$ obtained by removing the edges of $F$ from $G$.

\ah{For example, consider the graph in Figure~\ref{fig:good_edge}. We have rooted the graph at vertex $r$, and the vertices at the same distance from~$r$ are on the same level next to each other. Drawn this way, the horizontal edges with respect to~$r$ are horizontal and the vertical edges are either vertical or diagonal. For example, the edge~$v_1v_2$ is a horizontal edge since the distance from~$r$ is~$3$ to both of its endpoints. The vertex~$v_3$ is an endpoint to three vertical edges but no horizontal edges. The set~$B_r(u)$ consists of vertical edges whose endpoint is~$u$ and the other endpoint is a vertex on the level above~$u$. For example, $B_r(v_3)$ consists of the two vertical edges between $v_3$ and two of the neighbors of~$r$. The third vertical edge incident with $v_3$, $v_1v_3$, is not included in this set since $v_1$ is further away from $r$ and, in fact, we have $v_1v_3 \in B_r (v_1)$. Notice that the sets $B_r(u)$ form a partition of the vertical edges. The red edges form a good edge set with respect to~$r$ as it contains all horizontal edges and the required number of vertical edges.}

\begin{figure}
    \centering
    \begin{tikzpicture}
		\coordinate[label=left:$r \ $] (r) at (0,0);
		\coordinate (u1) at (-1,-1);
		\coordinate (u2) at (0,-1);
		\coordinate (u3) at (1,-1);
		\coordinate (u4) at (1,-2);
		\coordinate (u5) at (0,-2);
		\coordinate[label=left:$v_1 \ $] (v1) at (-.5,-3);
		\coordinate[label=right:$ \ v_2$] (v2) at (.5,-3);
		\coordinate[label=left:$v_3 \ $] (v3) at (-1,-2);
		\draw[thick] (r) -- (u1) -- (v3) -- (v1)   (u4) -- (u3) -- (r) -- (u2) -- (u5) -- (v2);
		\draw[thick, red] (v1) -- (u5) -- (u4) -- (v2) -- (v1)  (u2) -- (v3);
		\draw \foreach \x in {(r),(u1),(u2),(u3),(u4),(u5),(v1),(v2),(v3)} {
			\x node[circle, draw, fill=white,
			inner sep=0pt, minimum width=6pt] {}
		};
	\end{tikzpicture}
%
%
%
%
%
%
%
%
%
%
%
    \caption{\textcolor{black}{Vertices in the same horizontal line are equidistant from the root $r$. }The edges drawn in red indicate a set of good edges, \textcolor{black}{and these edges form a feedback edge set. Hence, a tree is obtained by removing the red edges from the given graph. }}
    \label{fig:good_edge}
\end{figure}
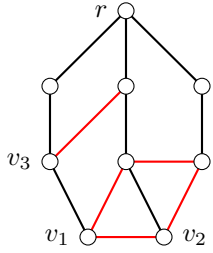

\begin{lemma}\label{lem:BFS}
For any connected graph $G$ with $n$ vertices and $m$ edges and a vertex $r\in V(G)$, \ah{a set $F$ of good edges} with respect to $r$ can be computed in $O(n+m)$ time.
\end{lemma}
\begin{proof}
By doing a Breadth First Search on $G$ from $r$, distances of $r$ from $u$ for all $u\in V(G)$ can be computed in $O(n+m)$ time. Then the horizontal and vertical edges can be computed in $O(m)$ time. Then the sets $B_r(u)$ for all $u\in V(G)$ can be computed in $O(n+m)$ time. Hence the set of good edges  with respect to $r$ can be computed in $O(n+m)$ time.
\end{proof}

\begin{lemma}\label{lem:tree}
For a set $F$ of good edges with respect to a vertex $r$ of a {connected} graph $G$, the subgraph $T_F$ is a tree rooted at $r$. Moreover, every path from $r$ to a leaf of $T_F$ is an isometric path in $G$.
\end{lemma}

\begin{proof}
First observe that $T_F$ is connected, as each vertex $u$ has exactly one edge $uv\in E(T_F)$ with $\dist{u}{r} = \dist{v}{r} + 1$. \ah{Indeed, removing horizontal edges or edges in $B_r(u)$ does not affect the sets $B_r(u')$ where $u' \neq u$.} Now assume for contradiction that $T_F$ has a cycle $C$. Let $v\in V(C)$ be a vertex that is furthest from $r$ among all vertices of $C$. Formally, $v$ is a vertex such that $\dist{r}{v} = \max \{ \dist{r}{w} \colon w\in V(C) \}$. Let $E'$ denote the set of edges in $T_F$ incident with $v$. Observe that $|E'|$ is at least two. Hence either $E'$ contains an horizontal edge, or $E' \cap B_r(v)$ contains at least two edges. Either case contradicts that $F$ is a good edge set with respect to $r$. This proves the first part of the observation. 
 
Now consider a path $P$ from $r$ to a leaf $v$ of $T_F$ and write it as $u_1 u_2\ldots u_k$ where $u_1=r$ and $u_k=v$.
By definition, we have $\dist{r}{u_i} = \dist{r}{u_{i-1}} + 1$ for each $i\in [2,k]$. Hence, $P$ is an isometric path in $G$. 
\end{proof}

\begin{observation}\label{obs:cyclomatic}
 Any set $F$ of good edges of a connected graph $G$ is a {feedback} edge set of $G$ with minimum cardinality.
\end{observation}

\begin{proof}
Due to Lemma~\ref{lem:tree} we have that $T_F$ is a tree and therefore $|F|=m-n+1$ which is same as the cardinality of a {feedback} edge set of $G$ with minimum cardinality.
\end{proof}


The \emph{base graph}~\cite{EpsteinWeighted} $G_b$ of a graph $G$ is the graph obtained from $G$ by iteratively removing vertices of degree~$1$ until there remain no such vertices. We use the base graph in some cases where preprocessing the leaves and other tree-like structures is needed.

\section{Metric dimension and variants}\label{sec:md}

In this section, we consider three metric dimension variants and conjectures regarding them and the cyclomatic number. We shall use the following result.

Distinct vertices $x,y$ are \emph{doubly resolved} by $v,u \in V(G)$ if $ \dist{v}{x} - \dist{v}{y} \neq \dist{u}{x} - \dist{u}{y} $. A set $S \subseteq V(G)$ is a \emph{doubly resolving set} of $G$ if every pair of distinct vertices of $G$ are doubly resolved by a pair of vertices in $S$. Lu et al.~\cite{LuWeightedDoubly} constructed a doubly resolving set of $G$ with $\delta (G) \geq 2$ by finding a good edge set with respect to a root $r \in V(G)$ using breadth-first search. We state a result obtained by  Lu et al.~\cite{LuWeightedDoubly} using the terminologies of this paper. 

\begin{theorem}[\cite{LuWeightedDoubly}]\label{thm:doubly}
	Let $G$ be a connected graph such that $\delta (G) \geq 2 $ and $r \in V(G)$. Let $S \subseteq V(G)$ consist of $r$ and the endpoints of the edges of a good edge set with respect to $r$.
	\begin{enumerate}[label=(\roman*)]
		\item The set $S$ is a doubly resolving set of $G$.
		\item If $r$ is a cut-vertex, then the set $S\setminus \{r\}$ is a doubly resolving set of $G$.
		\item We have  $|S| \leq 2 \cyclo{G} + 1$.
	\end{enumerate}
\end{theorem}

A doubly resolving set of $G$ is also a resolving set of $G$, and thus $\dim (G) \leq 2 \cyclo{G} + 1$ when $\delta (G) \geq 2$ due to Theorem~\ref{thm:doubly}. Moreover, if $G$ contains a cut-vertex and $\delta (G) \geq 2$, we have $\dim (G) \leq 2 \cyclo{G}$. \ah{Notice that $\delta (G) \geq 2$ implies that $L(G)=0$.} Therefore, Theorem~\ref{thm:doubly} implies that Conjecture~\ref{conj-dimedim} holds for the metric dimension of a graph with $\delta (G) \geq 2$ and at least one cut-vertex. \ah{When $\delta (G) \geq 2$ and the graph has no cut-vertex, we are $+1$ away from the bound in Conjecture~\ref{conj-dimedim}.}

A doubly resolving set is not necessarily an edge resolving set or a mixed resolving set. Thus, more work is required to show that edge and mixed resolving sets can be constructed with good edge sets.
A \emph{layer} of $G$ is a set $L_d=\{v\in V(G) \ | \ \dist{r}{v} = d\}$ where $r$ is the chosen root and $d$ is a fixed distance.

\begin{proposition}\label{prop:edgeresolving}
	Let $G$ be a graph with $\delta (G) \geq 2$, and let $r \in V(G)$. If the set $S$ contains $r$ and the endpoints of a good edge set $F$ with respect to $r$, then $S$ is an edge resolving set.
\end{proposition}
\begin{proof}
	Suppose to the contrary that there exist distinct edges $e=e_1e_2$ and $f=f_1f_2$ that are not resolved by $S$. In particular, we have $\dist{r}{e} = \dist{r}{f}$. Due to this, say, $e_1$ and $f_1$ are in the same layer $L_d$, and $e_2$ and $f_2$ are in $L_d \cup L_{d+1}$.
	If $e$ is a horizontal edge with respect to $r$, then $e_1,e_2 \in S$ and $e$ and $f$ are resolved. Thus, neither $e$ nor $f$ is a horizontal edge with respect to $r$ and we have $e_2,f_2 \in L_{d+1}$.
	
	If $e_2 = f_2$, then $e,f \in B_r(e_2)$. Thus, we have $e_2 \in S$ and at least one of $e_1$ and $f_1$ is also in $S$. Now $e$ and $f$ are resolved by $e_1$ or $f_1$. Therefore, we have $e_2 \neq f_2$ and $e_2,f_2 \notin S$.
	
	Let $w \in V(G)$ be a leaf in $T_F$ such that $e_2$ lies on a path between $w$ and $r$ in $T_F$. Since $\delta (G) \geq 2$, the vertex $w$ is an endpoint of some edge in $F$, and thus $w \in S$. Since $e$ and $f$ are not resolved by $S$, we have $\dist{w}{f_2} = \dist{w}{e_2} = d'-d-1$, where $w \in L_{d'}$, due to the path between $w$ and $r$ being isometric (Lemma~\ref{lem:tree}). 
	Let $P_f$ be \ah{an isometric} path $w - f_2$ in $G$, and assume that $P_f$ is such that it contains an element of $S$ as close to $f_2$ as possible. 
	Denote this element of $S$ by $s$. We have $s \in L_i$ for some $d+1 < i \leq d'$ (notice that we may have $s=w$). As the edges $e$ and $f$ are not resolved by $S$, we have $\dist{s}{e}=\dist{s}{f}$, which implies that $\dist{s}{e_2}=\dist{s}{f_2} = i-d-1$. Let $P_e'$ and $P_f'$ be \ah{isometric} paths $s-e_2$ and $s-f_2$, respectively. The paths $P_e'$ and $P_f'$ are internally vertex disjoint, since otherwise the vertex after which the paths diverge is an element of $S$ which contradicts the choice of $P_f$ and $s$. Let $v_e$ and $v_f$ be the vertices adjacent to $s$ in $P_e'$ and $P_f'$, respectively. 
	Now, we have $sv_e,sv_f \in B_r(s)$, and thus $v_e \in S$ (otherwise, $v_f \in S$, which contradicts the choice of $P_f$ and $s$). If $\dist{v_e}{e_2}<\dist{v_e}{f_2}$, then $v_e$ resolves $e$ and $f$, a contradiction. Thus, we have $\dist{v_e}{e_2} \geq \dist{v_e}{f_2}$, but now there exists \ah{an isometric} path $w - f_2$ that contains $v_e$, which is closer to $f_2$ than $s$ is, a contradiction.
\end{proof}

\begin{proposition}\label{prop:mixedresolving}
	Let $G$ be a graph with $\delta (G) \geq 2$, and let $r \in V(G)$. If the set $S$ contains $r$ and the endpoints of a good edge set $F$ with respect to $r$, then $S$ is a mixed resolving set.
\end{proposition}
\begin{proof}
	The set $S$ resolves all pairs of distinct vertices by Theorem~\ref{thm:doubly} and all pairs of distinct edges by Proposition~\ref{prop:edgeresolving}. Therefore we only need to show that all pairs consisting of a vertex and an edge are resolved.
	
	Suppose to the contrary that $v \in V(G)$ and $e=e_1e_2 \in E(G)$ are not resolved by $S$. In particular, the root $r$ does not resolve $v$ and $e$, and thus $v,e_1 \in L_d$ for some $d \geq 1$. If $e$ is a horizontal edge, then $e_1,e_2 \in S$ and $e$ and $v$ are resolved.
	Thus, assume that $e_2 \in L_{d+1}$. Let $w \in V(G)$ be a leaf in $T_F$ such that $e_2$ lies on a path between $w$ and $r$ in $T_F$. Since $\delta (G) \geq 2$, the vertex $w$ is an endpoint of some edge in $F$, and thus $w \in S$. We have $\dist{w}{e_2} = d'-d-1$, where $w \in L_{d'}$.
	However, now $\dist{w}{v} \geq d'-d > \dist{w}{e_2}$, and $w$ resolves $v$ and $e$, a contradiction.
\end{proof}

As pointed out in~\cite{LuWeightedDoubly}, if $R$ is a doubly resolving set that contains a cut-vertex $v$, then the set $R \setminus \{v\}$ is also a doubly resolving set. The following observation states that the same result holds for mixed resolving sets, and with certain constraints for (edge) resolving sets.

\begin{observation}\label{obs:cut}
	Let $G$ be a connected graph with a cut-vertex $v$.
	\begin{enumerate}[label=(\roman*)]
		\item Let $R \subseteq V(G)$ be such that there are at least two connected components in $G-v$ containing elements of $R$. If $\dist{v}{x} \neq \dist{v}{y}$ for some $x,y \in V(G) \cup E(G)$, then there exists an element $s \in R$, $s \neq v$, such that $\dist{s}{x} \neq \dist{s}{y}$.
		\item If $R \subseteq V(G)$ is a mixed resolving set of $G$, then every connected component of $G-v$ contains at least one element of $R$.
		\item If $R \subseteq V(G)$ is a resolving set or edge resolving set of $G$, then at most one connected component of $G-v$ does not contain any elements of $R$, and that component is isomorphic to $P_n$ for some $n \geq 1$.
	\end{enumerate}
\end{observation}
\begin{proof}
(i) If $s \in R$ is such that the \ah{isometric} paths from $s$ to $x$ and $y$ go through $v$, then we clearly have $\dist{s}{x} = \dist{s}{v} + \dist{v}{x} \neq \dist{s}{v} + \dist{v}{y} = \dist{s}{y}$.

Assume that there are only two connected components $G_1$, $G_2$ of $G-v$ that contain elements of $R$, and that $x \in V(G_1) \cup E(G_1 + v)$ and $y \in V(G_2) \cup E(G_2 + v)$. Let $s_1 \in V(G_1) \cap R$ and $s_2 \in V(G_2) \cap R$. Suppose that $\dist{s_1}{x} = \dist{s_1}{y}$. Now, $\dist{s_1}{v} + \dist{v}{y} = \dist{s_1}{x} \leq \dist{s_1}{v} + \dist{v}{x}$. Since $\dist{v}{x} \neq \dist{v}{y}$, we have $\dist{v}{y} < \dist{v}{x}$. Consequently, $\dist{s_2}{x} = \dist{s_2}{v} + \dist{v}{x} > \dist{s_2}{v} + \dist{v}{y} \geq \dist{s_2}{y}$. Thus, $x$ and $y$ are resolved by $s_2$.

(ii) If a connected component $G'$ does not contain an element of $R$, then $v$ and $vx$, where $x \in V(G')$ are not resolved by $R$, a contradiction.

(iii) Easy to see.
\end{proof}

The following corollary follows from Propositions~\ref{prop:edgeresolving} and~\ref{prop:mixedresolving}, and Observation~\ref{obs:cut}.

\begin{corollary}\label{cor:md:deg2}
	Let $G$ be a graph with $\delta (G) \geq 2$. 
	\begin{enumerate}[label=(\roman*)]
		\item If $G$ contains a cut-vertex, then $\edim (G) \leq 2 \cyclo{G}$ and $\mdim (G) \leq 2 \cyclo{G}$.
		\item If $G$ does not contain a cut-vertex, then $\edim (G) \leq 2 \cyclo{G} + 1$ and $\mdim (G) \leq 2 \cyclo{G} + 1$.
	\end{enumerate}
\end{corollary}

\ah{When $\delta(G) \geq 2$, both $L(G)=0$ and $\ell(G) = 0$. Thus, Corollary~\ref{cor:md:deg2} proves that when $\delta(G) \geq 2$, Conjecture~\ref{conj-dimedim} holds for the edge metric dimension when $G$ contains a cut-vertex and we are +1 away in the case where $G$ does not contain a cut-vertex. Compared to Conjecture~\ref{conj-mindeg2}, we are either +1 or +2 away from the conjectured bound for the edge metric dimension depending on whether $G$ contains a cut-vertex. As for the mixed metric dimension, Corollary~\ref{cor:md:deg2} proves that Conjecture~\ref{conj-mdim} holds when $\delta(G) \geq 2$ and $G$ contains a cut-vertex. If $\delta(G) \geq 2$ and the graph does not contain a cut-vertex, we are +1 away from the conjectured bound.}

We then turn our attention to graphs with $\delta(G) = 1$. We will show that a good edge set can be used to construct a (edge, mixed) resolving set also in this case. Moreover, we show that Conjecture~\ref{conj-mdim} holds, and Conjecture~\ref{conj-dimedim} holds when $\branchres{G} \geq 1$. We also show that $\dim(G)$ and $\edim (G)$ are at most $2 \cyclo{G} + 1$ when $\branchres{G}=0$. We use the following results on trees in our proof.

\begin{proposition}[\cite{KelencMixed}]\label{MMD:prop:trees}
	Let $T$ be a tree, and let $R \subseteq V(T)$ be the set of leaves of $T$. The set $R$ is a mixed metric basis of $T$.
\end{proposition}
\begin{proposition}[\cite{KelencEdge18,Khuller96}]\label{MDED:prop:trees}
	Let $T$ be a tree that is not a path. If $R \subseteq V(T)$ is a branch-resolving set of $T$, then it is a resolving set and an edge resolving set.
\end{proposition}

\begin{theorem}\label{thm:mindeg1}
	Let $G$ be a connected graph that is not a tree such that $\delta (G) = 1$. 
	Let $r \in V(G_b)$, and let $S \subseteq V(G_b)$ contain $r$ and the endpoints of a good edge set $F \subseteq E(G_b)$ with respect to $r$.
	If $R$ is a \ah{branch-resolving} set of $G$, then the set $R \cup S$ is a resolving set and an edge resolving set of $G$. If $R$ is the set of leaves of $G$, then the set $R \cup S$ is a mixed resolving set of $G$.
\end{theorem}
\begin{proof}
	Let $R$ be either a branch-resolving set of $G$ (for the regular and edge resolving sets) or the set of leaves of $G$ (for mixed metric dimension).
	We will show that the set $R \cup S$ is a (edge, mixed) resolving set of $G$.
	
	The graph $G-E(G_b)$ is a forest (note that some of the trees might be isolated vertices) where each tree contains a unique vertex of $G_b$. Let us denote these trees by $T_v$, where \ah{$v\in V(G_b) \cap V(T_v)$}.
	
	Consider distinct $x,y \in V(G) \cup E(G)$. We will show that $x$ and $y$ are resolved by $R \cup S$. 
	
	\begin{itemize}
		\item Assume that $x,y \in V(T_v) \cup E(T_v)$ for some $v \in V(G_b)$. Denote $R_v = (V(T_v) \cap R) \cup \{v\}$. The set $R_v$ is a (edge, mixed) resolving set of $T_v$ by Propositions~\ref{MDED:prop:trees} and~\ref{MMD:prop:trees}. If $x$ and $y$ are resolved by some element in $R_v$ that is not $v$, then we are done. If $x$ and $y$ are resolved by $v$, then they are resolved by any element in $S \setminus \{v\}$. Since $G$ is not a tree, the set $S \setminus \{v\}$ is clearly nonempty, and $x$ and $y$ are resolved in $G$.
	
		\item Assume that $x,y \in V(G_b) \cup E(G_b)$. Now $x$ and $y$ are resolved by $S$ due to Theorem~\ref{thm:doubly}, Proposition~\ref{prop:edgeresolving} or Proposition~\ref{prop:mixedresolving}.
	
		\item Assume that $x\in V(T_v) \cup E(T_v)$ and $y \in V(T_w) \cup E(T_w)$ where $v,w \in V(G_b)$, $v \neq w$. The set $S$ is a doubly resolving set of $G_b$ according to Theorem~\ref{thm:doubly}. Thus, there exist distinct $s,t \in S$ such that $\dist{s}{v}-\dist{s}{w} \neq \dist{t}{v}-\dist{t}{w}$. Suppose to the contrary that $\dist{s}{x}=\dist{s}{y}$ and $\dist{t}{x}=\dist{t}{y}$. Now we have
		\[ 	\dist{w}{y} - \dist{v}{x} = \dist{s}{v} - \dist{s}{w} \neq \dist{t}{v} - \dist{t}{w} = \dist{w}{y} - \dist{v}{x}, \]
		a contradiction. Thus, $s$ or $t$ resolves $x$ and $y$.

		\item Assume that $x\in V(T_v) \cup E(T_v)$ for some $v\in V(G_b)$, $v \neq x$, and $y = y_1y_2 \in E(G_b)$. Suppose that $\dist{r}{x} = \dist{r}{y}$. Without loss of generality, we may assume that $\dist{r}{y} = \dist{r}{y_1} = d$. Now $y_1 \in L_d$ and $v \in L_{d-d_x}$, where $d_x = \dist{v}{x} \in \{0, \ldots , d \}$. 
		If $y_2 \in L_d$, then $y$ is a horizontal edge and $y_1,y_2 \in S$. Now $x$ and $y$ are resolved by $y_1$ or $y_2$. So assume that $y_2 \in L_{d+1}$.
		Let $z \in V(G_b)$ be a leaf in $T_F$ such that $y_2$ lies on a path from $r$ to $z$ in $T_F$. Since $\delta (G_b) \geq 2$, the vertex $z$ is an endpoint of some edge in $F$, and thus $z \in S$. Now $z \in L_{d'}$ for some $d' > d+1$ and $\dist{z}{y_2} = d'-d-1$ by Lemma~\ref{lem:tree}.
		Consequently,
		\[ \dist{z}{x} = \dist{z}{v} + d_x \geq d' - (d-d_x) +d_x  = 2d_x + 1 + \dist{z}{y_2} > \dist{z}{y}.  \]
	\end{itemize}
\end{proof}

Since the root $r$ can be chosen freely, we can choose the root to be a cut-vertex in $G$ whenever $G$ contains cut-vertices. The bounds in the next corollary then follow from Observations~\ref{obs:cyclomatic} and~\ref{obs:cut}, and Theorem~\ref{thm:mindeg1}.

\begin{corollary}\label{cor:md:deg1}
	Let $G$ be a connected graph that is not a tree such that $\delta (G) = 1$. 
	We have $\dim (G) \leq \lambda(G) + 2 \cyclo{G}$, $\edim (G) \leq \lambda(G) + 2 \cyclo{G}$, and $\mdim (G) \leq \leaf{G} + 2 \cyclo{G}$,
\[ \dim (G) \leq \lambda(G) + 2 \cyclo{G}, \qquad \edim (G) \leq \lambda(G) + 2 \cyclo{G}, \qquad \mdim (G) \leq \leaf{G} + 2 \cyclo{G}, \]
	where $\lambda(G) = \max \{\branchres{G}, 1\}$.
\end{corollary}

The relationship of metric dimension and edge metric dimension has garnered a lot of attention since the edge metric dimension was introduced. Zubrilina~\cite{ZubrilinaEdge} showed that the ratio $\frac{\edim(G)}{\dim(G)}$ cannot be bounded from above by a constant, and Knor et al.~\cite{KnorEmdMd} showed the same for the ratio $\frac{\dim(G)}{\edim(G)}$. Inspired by this, Sedlar and \v{S}krekovski \cite{SedlarEdgedisjoint21} conjectured that for a graph $G \neq K_2$, we have $|\dim(G) - \edim(G)| \leq \cyclo{G}$. This bound, if true, is tight due to the construction presented in~\cite{KnorEmdMd}. It is easy to see that $\dim (G) \geq \lambda(G)$ and $\edim (G) \geq \lambda(G)$ (the fact that $\dim (G) \geq \branchres{G}$ is shown explicitly in~\cite{Chartrand00}, for example). Thus, we now obtain the bound $|\dim (G) - \edim (G)| \leq 2 \cyclo{G}$ due to the bounds established in Corollaries~\ref{cor:md:deg2} and~\ref{cor:md:deg1}.

\section{Geodetic sets and variants}\label{sec:geod}

We now address the problems related to geodetic sets, and show that the same method can be applied in this context as well. Note that all leaves of a graph belong to any of its geodetic sets. 

We first study geodetic sets explicitly, as this is the most well-studied notion in this area. We then address monitoring edge-geodetic sets and distance edge-monitoring sets, because these two notions have been studied in the literature in \ah{relation to the cyclomatic number and the number of leaves}. Bounds for other related parameters in Figure~\ref{fig:diagram} follow from the result on monitoring edge-geodetic sets.

 
\subsection{Geodetic sets}
 

 
 Let $\geod{G}$ denote the size of a smallest geodetic set of a graph $G$.
 
 \begin{theorem}\label{thm:geod}
 Let $G$ be a connected graph. If $G$ has a cut-vertex then $\geod{G} \leq 2\cyclo{G}+\leaf{G}$. Otherwise, $\geod{G} \leq 2\cyclo{G}+1$.
 \end{theorem}
 \begin{proof}
 Let $F$ be a set of good edges of $G$ with respect to a vertex $r$ obtained by Lemma~\ref{lem:BFS}. When $G$ has a cut-vertex then $r$ shall be a cut-vertex. Let $S$ be the union of the set of leaf vertices of $G$ and the endpoints of the edges in $F$. If $G$ is biconnected then include $r$ in $S$ also. Clearly, $|S| \leq 2\cyclo{G}+1+\leaf{G}$ when $G$ is biconnected and $|S| \leq 2\cyclo{G}+\leaf{G}$, otherwise. Due to Lemma~\ref{lem:tree}, $T_F$ is a tree and let $L$ be the set of leaf vertices of $T_F$. 
 
 If $G$ is biconnected, due to Lemma~\ref{lem:tree}, $L\cup \{r\}$ is a geodetic set of $G$, and clearly $(L \cup \{r\}) \subseteq S$. Therefore $S$ is a geodetic set of $G$. 
 
 Otherwise, $r$ is a cut-vertex and $L\subseteq S$. Consider any vertex $u$ of $G$. Let $C_u$ denote the connected component of $G-\{r\}$ containing $u$. Due to Lemma~\ref{lem:tree}, $u$ lies in an isometric path between $r$ and a vertex $u'\in V(C_u) \cap L$. Since $r$ is a cut-vertex there exists a connected component $C'$ different from $C_u$ and $L\cap V(C')$ is non-empty. Let $u''$ be a vertex of $L\cap V(C')$. Clearly $\dist{u'}{u''} = \dist{r}{u'} + \dist{r}{u''}$. Let $P'$ be the path in $T_F$ between $r$ and $u'$ and $P''$ be the path in $T_F$ between $r$ and $u''$. We have that $P=P'\cup P''$ is an isometric path in $G$ and contains $u$. Hence, the set $L$ is a geodetic set of $G$, $S$ is a geodetic set of $G$. 
 \end{proof}
 
The upper bound of Theorem~\ref{thm:geod} is tight when there is a cut-vertex, indeed, consider the graph formed by a disjoint union of $k$ odd cycles and $l$ paths, all identified via a single vertex. The obtained graph has cyclomatic number $k$, $l$ leaves, and geodetic number $2k+l$. Observe that any odd cycle has geodetic number~3 and cyclomatic number~1, so the bound \ah{$\geod{G} \leq 2\cyclo{G}+1$ is tight when there is no cut-vertex in the graph}. \ff{However, we do not know the best possible bound for 2-connected graphs with arbitrarily large values of $c(G)$. The graph $G_k$ obtained from $K_{2,k}$ by adding an edge between the two vertices of degree~$k$ is 2-connected, has cyclomatic number $k$ and geodetic number $k$. We do not know a family of 2-connected graphs of arbitrarily large cyclomatic number whose geodetic number is closer to twice the cyclomatic  number and leave this as an open problem.} 

\subsection{Monitoring edge-geodetic sets}

Let $\meg{G}$ denote the size of a smallest monitoring edge-geodetic set of a graph $G$. It was proved in~\cite{MEG1} that $\meg{G}\leq 9\cyclo{G}+\leaf{G}-8$ for every graph $G$, and some graphs were constructed for which $\meg{G}=3\cyclo{G}+\leaf{G}$. We next improve the former upper bound, therefore showing that the latter construction is essentially best possible.


\begin{theorem}\label{thm:meg}
For any graph $G$, we have $\meg{G}\leq 3\cyclo{G}+\leaf{G}+1$. If $G$ contains a cut-vertex, then $\meg{G}\leq 3\cyclo{G}+\leaf{G}$. 
\end{theorem}
\begin{proof}
It is not hard to check that when $G$ is a tree, we have $\meg{G}=\leaf{G}$~\cite{MEG1}. Let us thus assume that $\cyclo{G} \geq 1$.

We will first show the bound $\meg{G}\leq 3\cyclo{G}+\leaf{G}+1$. We say that an edge $e$ is \emph{monitored} by a set $S$ if there are two vertices $x,y$ of $S$ such that $e$ lies on all \ah{isometric} paths between $x$ and $y$; $S$ is thus a monitoring edge-geodetic set of $G$ if it monitors all edges of $G$.

Let $r$ be an arbitrary vertex of $G$ that belongs to some cycle of $G$, and let $F$ be a good set of edges (with respect to $r$) obtained through Lemma~\ref{lem:BFS}. We construct a set $S$ as follows: $r$ belongs to $S$, all leaves of $G$ belong to $S$, and for each edge of $F$, both its endpoints belong to $S$. Moreover, for each vertex $u$ of $G$ with $|B_r(u)|\geq 2$, we add to $S$ all endpoints of the edges of $B_r(u)$ (not just the $|B_r(u)|-1$ ones that are in $F$).

It is clear that we have $|S|\leq 3|F|+\leaf{G}+1$, since we add to $S$ at most two vertices for each edge in $F$, with an additional vertex whenever there is a vertical edge in $F$ (this can only happen $|F|$ times). By Observation~\ref{obs:cyclomatic}, we have $|F|\leq \cyclo{G}$, and so $|S|\leq 3\cyclo{G}+\leaf{G}+1$.

We now show that $S$ is a monitoring edge-geodetic set. Let $e=uv$ be an edge of $G$. If both endpoints $u,v$ of $e$ are in $S$, then $e$ is clearly monitored by them. In particular, this is the case if $e$ is horizontal with respect to $r$.

Assume that $e$ has at most one endpoint in $S$. Thus $e$ is vertical with respect to $r$, and $e$ is also an edge in $T_F$. Let us consider the tree $T_F$ further, and remove the additional edges of the sets $B_r(w)$ whose endpoints we added to $S$, i.e. all edges in some $B_r (w)$ where $|B_r(w)|\geq 2$. The resulting graph is a forest where each tree can be rooted in a natural way by following the structure of $T_F$ when the root is $r$. Of these trees, every root and every leaf is in $S$. Moreover, the edge $e$ is an edge in one such tree. Consider a leaf $l$ and a root $r'$ of one tree $T$. Now there is in fact a unique \ah{isometric} path between $l$ and $r'$ in $G$, and this path is completely contained in $T$ (otherwise, there is a vertical edge whose endpoints we added to $S$ in the path between $l$ and $r'$ in $T$, and $l$ is not in the same tree as $r'$, a contradiction). Thus, $e$ is monitored by a leaf and the root of the tree that contains $e$.

Assume then that $G$ contains a cut-vertex. Since the root $r$ can be arbitrarily, we can choose $r$ to be a cut-vertex. Now the set $S \setminus \{r\}$ is a monitoring edge-geodetic set. Indeed, if $r$ and $u \in S$ monitor an edge $e$, then $e$ is also monitored by $u$ and $v \in S$ where $v$ is in a different connected component of $G - r$ as $u$. Thus, when $G$ contains a cut-vertex, we have $\meg{G}\leq 3\cyclo{G}+\leaf{G}$.
\end{proof}


\subsection{Distance-edge-monitoring-sets}

We now turn our attention to distance-edge-monitoring-sets and \ah{prove} Conjecture~\ref{conj-dem}.  Let $\dem{G}$ denote the size of a smallest distance-edge-monitoring set of a graph $G$. We shall use the following lemma.

\begin{lemma}[{\cite[Observation 4]{DEM1}}]\label{lem:DEM}
For any graph $G$, any distance-edge-monitoring-set of its base graph $G_b$ is also one of $G$.
\end{lemma}

\begin{theorem}\label{thm:dem}
For any connected graph $G$, $\dem{G}\leq \cyclo{G}+1$.
\end{theorem}
\begin{proof}
We say that an edge $e$ is \emph{monitored} by a set $S$ if there are two vertices $x\in S$ and $y\in V(G)$ such that $e$ lies on all \ah{isometric} paths between $x$ and $y$; $S$ is thus a monitoring edge-geodetic set of $G$ if it monitors all edges of $G$.

The statement is true when $\cyclo{G}\leq 2$ by~\cite{DEM1}, so we can assume $G$ is not a tree. By Lemma~\ref{lem:DEM}, we may assume $G$ has minimum degree at least~2.

Let $r$ be an arbitrary vertex of $G$, and let $F$ be a good set of edges (with respect to $r$) obtained through Lemma~\ref{lem:BFS}. We construct a set $S$ as follows: $r$ belongs to $S$, and for each edge of $F$, one of its endpoints belongs to $S$; if the edge is vertical with respect to $r$, we choose the endpoint that is farthest from $r$; if it is horizontal, we choose any of the two endpoints.

It is clear that $|S|\leq \cyclo{G}+1$, as $|S|\leq |F|+1$ and $|F|\leq \cyclo{G}$ by Observation~\ref{obs:cyclomatic}. We next show that $S$ is distance-edge-monitoring.

Let $e=uv$ be an edge of $G$. If $e$ is an edge of $F$, it is monitored by $S$, since one of its endpoints is in $S$. Thus, we can  assume that $e$ is not in $F$. Hence, it is vertical with respect to $r$, and we assume without loss of generality that $\dist{u}{r}=\dist{v}{r}+1$. Since $e$ is not in $F$, we have $|B_r(u)|=1$. Thus, all \ah{isometric} paths from $u$ to $r$ go through $e$, and hence $e$ is monitored by $S$. Thus all edges are monitored by $S$, which establishes the claim.
\end{proof}

\section{Path covers/partitions and variants}\label{sec:paths}

In this section, we consider the path covering and partition problems. We focus on isometric path edge-covers (sets of isometric paths that cover all edges of the graph) and on isometric path partitions (sets of isometric paths that partition the vertex set of the graph). Indeed, those have the most restrictive definitions and the obtained bounds thus hold for the other related path covering/partition problems from Figure~\ref{fig:diagram}. We denote by $\ipec{G}$ and by $\ipp{G}$, respectively, the smallest size of an isometric path edge-cover and isometric path (vertex-)partition, respectively.

\subsection{Isometric path edge-cover}

\begin{theorem}\label{thm:ipec}
For any graph $G$, $\ipec{G}\leq 3\cyclo{G}+\lceil(\leaf{G}+1)/2\rceil$.
\end{theorem}
\begin{proof}
We construct a set $S$ of paths as follows. Consider the base graph $G_b$ of $G$, let $r$ be an arbitrary vertex of $G_b$, and let $F$ be a good set of edges of $G_b$ (with respect to $r$) obtained by Lemma~\ref{lem:BFS}. 
For each horizontal edge $xy$ of $F$, we add it (as a path) to $S$, as well as \ah{an isometric} path from $x$ to $r$ and one from $y$ to $r$. For each vertex $v$ with $|B_r(v)|\geq 2$, we add to $S$, $|B_r(v)|$ \ah{isometric} paths from $v$ to $r$, each starting with a different edge from $B_r(v)$. Finally, if $G$ contains some leaves, we consider the set $L'$ formed by the set $L$ of leaves of $G$, together with the vertex $r$, and an arbitrary pairing of $L'$, of size $\lceil(\leaf{G}+1)/2\rceil$, such that each vertex of $L'$ is paired with some other one (if $\leaf{G}+1$ is odd, one vertex of $L'$ may be paired with two vertices, but all others are paired with one only). 

Consider a set $S'$ of $\lceil(\leaf{G}+1)/2\rceil$ \ah{isometric} paths between each two mutually paired vertices of $L'$. We incrementally modify the pairing as follows. If all edges of $G-E(G_b)$ are covered by $S'$, we do nothing. Otherwise, assume some edge $e$ of $G-E(G_b)$ is not covered by a path of $S'$. Note that $G-E(G_b)$ is a forest, and $e$ is a bridge of $G$, with one component of $G-e$ containing at least one leaf, and one component containing $r$. Since that leaf and $r$ are paired vertices in $L'$ but $e$ is not covered by the \ah{isometric} paths in $S'$, there must exist two pairs $\{a,b\}$ and $\{c,d\}$ of paired vertices of $L'$, each pair being in one component of $G-e$, say $a$ and $b$ are in a tree component of $G-e$ and $c,d$ are in the same component as $r$. We modify the pairing of $L'$ by replacing $\{a,b\}$ and $\{c,d\}$ by $\{a,c\}$ and $\{b,d\}$, and we claim that the \ah{isometric} paths induced by this new pairing cover the same edges from $G-E(G_b)$ as before, and also, the edge $e$. Indeed, as $G-E(G_b)$ is a forest (and each component of $G-E(G_b)$ has a unique vertex with neighbors in $G_b$), any \ah{isometric} path (in $G$) from a leaf to any vertex of the same component of $G-E(G_b)$ is unique. Thus, the $a-b$ \ah{isometric} path is unique, and any \ah{isometric} path from $a$ to $c$ (and from $b$ to $d$) goes through $e$. Hence, the union of any two such \ah{isometric} paths contains the edges of the $a-b$ path. A similar argument holds for the path from $c$ to $d$. Thus, we continue this process until all edges of $G-E(G_b)$ are covered by $S'$, increasing the number of covered edges at each step. Finally, we add $S'$ to $S$. 

It is clear that $S$ contains only isometric paths, by construction; moreover, $|S|\leq 3\cyclo{G}+\lceil(\leaf{G}+1)/2\rceil$ because we add at most three paths to $S$ for each edge of $F$ in the first steps of the construction, and the last step of the construction adds $\lceil(\leaf{G}+1)/2\rceil$ additional paths to $S$. It remains to show that $S$ covers all edges of $G$. Let $e$ be an edge. If $e$ is an edge of $G-E(G_b)$, by the last part of the construction of $S$, $e$ is covered by some \ah{isometric} path between two vertices of $L'$. If $e$ is a horizontal edge with respect to $r$, then $e$ itself is a path of $S$, so $e$ is covered. If $e$ is a vertical edge with respect to $r$ and $e$ belongs to $G_b$, then there must be a vertex $w$ of $G_b$ with $|B_r(w)|\geq 2$ and $e$ lies on \ah{an isometric} path from $w$ to $r$. Let $w$ be chosen so as to be the closest to $e$, among all such vertices. Then, we have selected some \ah{isometric} path from $w$ to $r$ in the first step of the construction, each containing a distinct edge of $B_r(w)$. By the choice of $w$, one such path from $w$ to $r$ goes through $e$, and thus $e$ is covered. 
\end{proof}

The upper bound of Theorem~\ref{thm:ipec} is nearly tight, indeed, consider (again) the graph formed by a disjoint union of $k$ odd cycles and $l$ paths, all identified via a single vertex. The obtained graph has cyclomatic number $k$, $l$ leaves, and isometric path edge-cover number $3k+\lceil l/2\rceil$.

\subsection{Isometric path partition}

\begin{theorem}\label{thm:ipp}
For any graph $G$, $\ipp{G}\leq 2\cyclo{G}+\leaf{G}$.
\end{theorem}
\begin{proof}
The proof is similar to the one of Theorem~\ref{thm:ipec}, except that it is not necessary to consider the base graph. Consider an arbitrary vertex $r$ of $G$, and let $F$ be a good set of edges of $G$ (with respect to $r$) obtained by Lemma~\ref{lem:BFS}. We construct a set $S$ of isometric vertex-disjoint paths as follows. Initially, $S=\emptyset$. For each horizontal edge $xy$ of $F$, we consider \ah{an isometric} path $P_x$ from $x$ to $r$ and \ah{an isometric} path $P_y$ from $y$ to $r$. We first add to $S$ the maximal subpath of $P_x$ (starting at $r$), that does not overlap with any path already in $S$. Then, we do the same for $P_y$. Similarly, for each vertex $v$ with $|B_r(v)|\geq 2$, we let $N_v$ be the set of vertices $u$ with $uv\in B_r(v)$. We consider an arbitrary \ah{isometric} path from $v$ to $r$ (it goes through some vertex, say $u_0$, of $N_v$). Then, for every vertex $u\in N_v\setminus\{u_0\}$, we sequentially consider \ah{an isometric} path $P_u$ from $u$ to $r$. We proceed as before for each of these paths: we add to $S$ a maximal subpath of $P_x$ starting at $x$, that does not intersect any path already in $S$. Finally, for every leaf $v$ of $G$, we again consider \ah{an isometric} path from $v$ to $r$, and select its maximal subpath (starting from $v$) that does not intersect any other path in $P$.

We have added at most $2\cyclo{G}+\ell(G)$ paths to $S$ (at most two for every edge in $F$ and at most one for every leaf). By construction, $S$ contains only isometric paths, and they are pairwise vertex-disjoint. By similar arguments as in the proof of Theorem~\ref{thm:ipec}, every vertex is covered by a path of $S$, and so we have obtained an isometric path partition of the desired size.
\end{proof}

Once again, the bound is nearly tight by the same graph formed by a disjoint union of $k$ odd cycles (each of length at least 5) and $l$ paths, all identified via a single vertex. It has cyclomatic number $k$, $l$ leaves, and isometric path partition number $2k+l-1$.

We remark that for many graphs, the bound of $\ell(G)$ to cover the leaves of the graph in the above theorem is not necessarily tight. Indeed, for some trees, one may have an (isometric) path partition of size $\lceil\leaf{G}/2\rceil$: consider for example a path $P_t$ on $t\geq 3$ vertices and for each of its $t-2$ internal vertices, say $v$, attach a copy of $P_3$ whose central vertex is made adjacent to $v$. This tree has $2t-2$ leaves and a path partition of size $t-1$.

In fact, it is known that one gets a more precise (always tight) bound for the path partition number of a tree by considering its \emph{scattering number}, see~\cite{JUNG}.

\section{Algorithmic consequences}\label{sec:algo}

As mentioned in the introduction, for many of the problems studied here, there are no efficient algorithms for graphs of bounded treewidth: for example, it is NP-hard to compute the metric dimension or the geodetic number of graphs of constant treewidth (in fact, even pathwidth)~\cite{LP22,tale2025GS}. 
The complexity of the strong geodetic set problem with respect to treewidth is unknown~\cite{DFPT24}.

Recall that the treewidth is at most the cyclomatic number plus one. Next, we show that for the vertex-subset problems studied in this paper, our bounds imply polynomial-time (XP) algorithms for graphs of bounded cyclomatic number, thus partially answering some algorithmic open problems in this area.

\begin{theorem}\label{thm:algo}
For all the variants of geodetic sets and metric dimension considered in this paper, if we have an upper bound on the solution size of $a\cdot\cyclo{G}+f(\leaf{G})$ for some $a\in\mathbb{N}$, we obtain an algorithm with running time $O(n^{a\cdot\cyclo{G}})$ on graphs $G$ of order $n$.
\end{theorem}
\begin{proof}
The algorithm needs to pre-process the leaves and compute a subset of the leaves of size $f(\leaf{G})$. This can be done in polynomial time for all the considered problems. For geodetic set types of problems, one simply selects all the leaves (for distance-edge-monitoring sets, we must not select any leaf). This can be done in time $O(n+\cyclo{G})=O(n^2)$. For metric dimension related problems, one has to compute the structural shape of the leaves; this can be done in time $O(n+\cyclo{G})$ as well, see for example~\cite{EpsteinWeighted,Khuller96}. 

After that (and noting that selecting $f(\leaf{G})$ leaves of $G$ is necessary in each of the considered problems), the proofs of our bounds show that $a\cdot\cyclo{G}$ are sufficient to extend the chosen leaf subset to a solution. It thus suffices to iterate over all possible subsets of vertices of size at most $a\cdot\cyclo{G}$:  consider this as a potential solution, and add the required set of leaves to the solution, and check whether it is a valid solution. This yields the desired running time.
\end{proof}

\section{Conclusion}\label{sec:conclu}

We have demonstrated that a simple technique based on breadth-first-search is very efficient to obtain bounds for many distance-based covering problems, when the cyclomatic number and the number of leaves are considered. This resolves or advances several open problems and conjectures from the literature on this type of problems. There remain some gaps between the obtained bounds and the conjectures or known constructions, that still need to be closed. \ff{Moreover, in some cases, obtaining the best possible upper bound for 2-connected graphs of arbitrarily large cyclomatic number is also an interesting open problem, for example for the geodetic number and the monitoring-edge geodetic number.}

A refinement of the cyclomatic number of a (connected) graph $G$ is called its \emph{max leaf number}, which is the maximum number of leaves in a spanning tree of $G$. It is known that the cyclomatic number is always upper-bounded by a quadratic function of the max leaf number plus the number of leaves~\cite{E15}, so, all our bounds also imply bounds using the max leaf number only.

Regarding the algorithmic applications, we note that the XP algorithms described in Theorem~\ref{thm:algo} can sometimes be improved to obtain an FPT algorithm. This is the case for geodetic sets~\cite{KeKo20}, but whether this is possible for the metric dimension remains a major open problem~\cite{E15,KeKo20} (this is however shown to be possible for the larger parameter ``max leaf number''~\cite{E15}). 

Also, we leave open whether similar algorithms can be obtained for the path covering/partition problems studied here. As one has to decide where the paths go through, it seems a little bit more difficult to design such an algorithm than for the vertex-subset problems. This would be quite interesting, as the complexity of the isometric path cover problem with respect to treewidth is unknown~\cite{DFPT24}.

We have not been exhaustive. There might be other distance-based covering problems for which the same approach can be used. For example, this is the case for the problem of \emph{geodesic-transversal}~\cite{MBK22} (also called \emph{maximal shortest path cover} in~\cite{PS21}): a set $S$ of vertices of a graph $G$ such that every maximal \ah{isometric} path of $G$ contains a vertex of $S$.

Another distance-related problem which has been studied in relation with the cyclomatic number is the \emph{edge-tracking path problem}~\cite{TrackingPaths}, a variant of the more studied \emph{(vertex)-tracking path problem}~\cite{BKPS20,TrackingPaths,EGLM19}. It is shown in~\cite{TrackingPaths} that the edge-tracking path problem can be reduced to essentially finding a feedback edge set of the graph. This result also implies that the size of an optimal (vertex)-tracking path set is at most twice the cyclomatic number. Note that, however, the tracking path problem is known to be efficiently solvable for graphs of bounded treewidth~\cite{EGLM19}, hence it behaves differently from most problems studied here.

For other distance-based parameters, the type of bounds that are studied here do not hold. For example, this is the case for \emph{strong metric dimension}, that is more constrained than the metric dimension: a solution requires half of the vertices for any cycle graph~\cite{OP07} (which has cyclomatic number 1).



\bibliographystyle{abbrv}
\bibliography{references}

\end{document}